\documentclass[11pt]{article}

\usepackage{amsmath,amssymb,amsthm}
\usepackage{graphicx}
\usepackage{bm}
\usepackage{hyperref}\hypersetup{hidelinks}

\title{\bf
Relative Entropy Revivals Caused by Memory-Induced Loss of Divisibility in Classical Non-Markovian Dynamics
}

\author{Koichi Nakagawa\\
Hoshi University, Tokyo, Japan}
\date{}

\theoremstyle{plain}
\newtheorem{proposition}{Proposition}

\newtheorem{criterion}{Heuristic criterion}

\begin{document}
\maketitle

\begin{abstract}
Relative-entropy revivals and negative entropy-production rates are established
signatures of memory effects in non-Markovian dynamics.  The present paper
addresses a more specific question: what minimal dynamical mechanism produces
entropy overshoot in classical stochastic relaxation?  We define entropy
overshoot as a transient increase of the Kullback--Leibler divergence from a
stationary distribution during a relaxation process that ultimately returns to
equilibrium.

Starting from finite-state generalized master equations with memory kernels, we
analyze the effective time-local generator governing the observed dynamics.  The
standard monotonic decay of relative entropy is recalled as a reference property
of reversible Markov dynamics.  We then show that memory can destroy stochastic
divisibility by producing transiently negative effective rates.  Once the
evolution ceases to be divisible into intermediate Markov steps, the usual
relative-entropy contraction argument no longer applies, and entropy overshoot
can occur.

The contribution of the paper is therefore not the introduction of another
witness of non-Markovianity, but the identification of a minimal structural
mechanism and its phase structure.  The mechanism is illustrated with two- and
three-state stochastic models, where the competition between the intrinsic
relaxation time and the memory time separates monotonic relaxation from
overshoot.  Numerical phase diagrams show that the overshoot region expands when
memory persists on the relaxation timescale.

These results clarify the relation between classical information backflow,
negative entropy-rate behavior, and memory-induced loss of divisibility.  They
also distinguish the present mechanism-based interpretation from previous work
that used relative-entropy revivals or entropy-production rates primarily as
indicators of non-Markovianity.
\end{abstract}


\section{Introduction}

Relaxation toward equilibrium is one of the central themes of nonequilibrium
statistical physics.  In a finite-state Markov process satisfying detailed
balance, the probability distribution approaches its stationary distribution in
a monotonic way when measured by the Kullback--Leibler (KL) divergence
\cite{CoverThomas2006}.  The KL divergence then plays the role of a Lyapunov
function and quantifies the decrease of information distance from equilibrium
\cite{Qian2001,Seifert2012}.

This familiar picture can fail when memory effects are present. Non-Markovian
dynamics arise naturally in projection-operator approaches introduced by
Nakajima and Zwanzig \cite{Nakajima1958,Zwanzig2001}, as well as in generalized
Langevin equations and anomalous transport processes \cite{Kubo1966,MetzlerKlafter2000}.
They are also closely related to continuous-time random walks and semi-Markov
constructions \cite{MontrollWeiss1965,VanKampen2007}.  In such systems the
future evolution is not determined solely by the instantaneous probability
distribution, but also by the past history of the process.

Several indicators of non-Markovianity have been developed in quantum and
classical settings.  These include trace-distance revivals, relative-entropy
revivals, entropy-production rates, and divisibility-based criteria
\cite{Breuer2009,Rivas2014,Vacchini2012,BreuerPetruccione2002}.  In particular,
the relationship between classical and quantum notions of non-Markovianity was
analyzed in Ref.~\cite{Vacchini2012}, while negative entropy-production rates
and their relation to memory effects have been discussed in Refs.~\cite{Bhattacharya2017,Popovic2018,Strasberg2019}.  Recent work has also explored
experimental and applied contexts in which transiently negative entropy
production serves as a signature of memory \cite{Magnetization2026}.

The purpose of the present paper is not to introduce another witness of
non-Markovianity.  Rather, it is to isolate a minimal mechanism behind a
specific transient behavior: a temporary increase of the KL divergence during
non-Markovian relaxation.  We call this phenomenon entropy overshoot.  Although
relative-entropy revivals and negative entropy-production rates have been used
before as indicators of memory, the mechanism that produces such an overshoot
can be obscured in general models.  Our aim is therefore to make the mechanism
as transparent as possible in finite-state stochastic systems.

The main message of this work is structural.  We use the standard Markovian
monotonicity of relative entropy as a reference point and ask how it can fail.
We show that memory kernels can generate an effective time-local generator with
transiently negative off-diagonal elements.  These negative effective rates do
not make the probability distribution ill-defined.  Instead, they signal a loss
of stochastic divisibility: the observed dynamics can no longer be decomposed
into intermediate Markov steps.  Once this divisibility is lost, the usual
relative-entropy contraction argument does not apply, and the KL divergence can
increase for suitable initial conditions.

This viewpoint connects entropy overshoot to information backflow, but it also
separates the present contribution from earlier indicator-based approaches.  The
central object here is not the revival itself; it is the memory-induced loss of
divisibility that generates the revival.  Minimal two- and three-state models
are used to display this mechanism explicitly and to construct phase diagrams
separating monotonic relaxation from overshoot.

The paper is organized as follows.  Section~2 introduces generalized master
equations with memory kernels.  Section~3 discusses stationary distributions,
and Sec.~4 recalls the relative-entropy measure used throughout the paper.
Sections~5--8 formulate entropy overshoot, Markovian embedding, and effective
time-local generators.  Section~9 states the standard monotonicity property of
relative entropy under Markovian dynamics, while Sec.~10 explains the mechanism
by which non-divisibility can break this monotonicity.  Section~11 gives a
simple timescale criterion for overshoot.  Sections~12--16 present minimal
models, numerical methods, and divisibility-loss phase diagrams.  Section~17
discusses the relation to previous entropy-production and information-backflow
approaches, and Sec.~18 summarizes the conclusions.

\section{Generalized master equations}

Consider a finite-state stochastic system with state space
\begin{equation}
  X=\{1,2,\ldots,N\}.
\end{equation}
Let $P_i(t)$ be the probability of state $i$ at time $t$, and let
$P(t)=(P_1(t),\ldots,P_N(t))^T$.  A broad class of non-Markovian stochastic
dynamics can be described by a generalized master equation
\begin{equation}
  \frac{d}{dt}P_i(t)=\sum_j\int_0^t K_{ij}(t-\tau)P_j(\tau)\,d\tau,
  \label{eq:gme_component}
\end{equation}
or, in vector form,
\begin{equation}
  \frac{d}{dt}P(t)=\int_0^t K(t-\tau)P(\tau)\,d\tau.
  \label{eq:gme}
\end{equation}
Here $K(t)$ is a memory kernel.  The Markovian master equation is recovered when
the kernel becomes local in time.

A simple and analytically tractable example is the exponential kernel
\begin{equation}
  K(t)=\gamma e^{-\gamma t}R,
  \label{eq:exp_kernel}
\end{equation}
where $R$ is a Markov generator and $\gamma$ controls the memory time.  Large
$\gamma$ corresponds to short memory, while small $\gamma$ corresponds to long
memory.  In the formal short-memory limit the dynamics approaches the Markovian
relaxation generated by $R$.

The exponential kernel is not the most general possible memory kernel, but it is
sufficient to display the basic mechanism studied in this paper.  It also admits
a convenient Markovian embedding, which makes the relation between memory and
loss of divisibility particularly transparent.

Non-Markovian stochastic dynamics have been studied extensively in classical statistical physics 
\cite{VanKampen2007,HanggiJung1995}. In particular, generalized master equations with memory kernels 
provide a natural framework to describe systems with temporal correlations.

\section{Stationary distribution}

We assume that the Markov generator $R$ admits a stationary distribution $\pi$
satisfying
\begin{equation}
  R\pi=0,
  \label{eq:stationary}
\end{equation}
with $\pi_i>0$ and $\sum_i\pi_i=1$.  For an irreducible finite-state Markov
chain this stationary distribution is unique.

For the exponential kernel~\eqref{eq:exp_kernel}, the same stationary
distribution is also stationary for the non-Markovian dynamics.  Indeed, if
$P(t)=\pi$ for all $t$, then
\begin{equation}
  \int_0^t \gamma e^{-\gamma(t-\tau)}R\pi\,d\tau=0.
\end{equation}
Thus the memory kernel changes the transient relaxation, but not the final
stationary state.  This makes $\pi$ a natural reference state for the relative
entropy.

\section{Relative entropy and entropy overshoot}

We measure the distance from the stationary state by the KL divergence \cite{CoverThomas2006}
\begin{equation}
  F(t)=D_{\rm KL}(P(t)\|\pi)
      =\sum_i P_i(t)\ln\frac{P_i(t)}{\pi_i}.
  \label{eq:KL}
\end{equation}
For Markov processes satisfying detailed balance, $F(t)$ is monotonically
non-increasing \cite{Qian2001,Seifert2012}.  This monotonicity expresses the irreversible loss of information
about the initial state during relaxation.

We define entropy overshoot as the occurrence of
\begin{equation}
  F(t)>F(0)
  \label{eq:overshoot_definition}
\end{equation}
for some finite time $t>0$, even though $F(t)$ eventually decreases to zero.
This definition is intentionally operational: it identifies a transient
increase of information distance from equilibrium.  The mechanism behind this
increase is the main subject of the paper.

It is useful to distinguish entropy overshoot from ordinary slow relaxation.
In slow but monotonic relaxation, $F(t)$ decreases at all times.  In entropy
overshoot, the initial trend can be reversed by memory-induced feedback before
asymptotic relaxation takes over.

\section{Embedded Markov representation}

For the exponential kernel one can rewrite the non-Markovian equation as a
Markovian system in an enlarged state space.  Introduce an auxiliary vector
\begin{equation}
  Y(t)=\int_0^t \gamma e^{-\gamma(t-\tau)}P(\tau)\,d\tau.
\end{equation}
Then
\begin{align}
  \dot P(t)&=RY(t),\label{eq:embed_P}\\
  \dot Y(t)&=\gamma P(t)-\gamma Y(t).\label{eq:embed_Y}
\end{align}
Equations~\eqref{eq:embed_P}--\eqref{eq:embed_Y} form a time-local system.  The
original probability vector $P(t)$ is obtained by projecting the enlarged state
onto the observable component.  The embedding shows explicitly that the
observable process can be non-Markovian even when the enlarged dynamics is
local in time.

This representation is a useful diagnostic tool: memory corresponds to degrees
of freedom that are not included in the observed probability vector.  Delayed
feedback from these hidden variables can temporarily reverse the direction of
relaxation in the observed space.

\section{Effective time-local generator}

Whenever the map from $P(0)$ to $P(t)$ is invertible, the observable dynamics can
be written in time-local form
\begin{equation}
  \frac{d}{dt}P(t)=R_{\rm eff}(t)P(t),
  \label{eq:time_local}
\end{equation}
where $R_{\rm eff}(t)$ is an effective generator.  If $R_{\rm eff}(t)$ is a valid
Markov generator at every time, it satisfies
\begin{align}
  (R_{\rm eff})_{ij}(t)&\geq 0\qquad (i\ne j),\label{eq:positive_rates}\\
  \sum_i (R_{\rm eff})_{ij}(t)&=0.\label{eq:column_sum}
\end{align}
In that case the evolution is divisible: the map from any time $s$ to any later
time $t$ is again a stochastic map.

Non-Markovian memory can produce an effective generator that violates
Eq.~\eqref{eq:positive_rates} during finite time intervals.  Such negative
effective rates do not mean that probabilities become negative.  Rather, they
mean that the reduced observable dynamics cannot be decomposed into ordinary
Markovian steps at those times.

\section{Evolution of the KL divergence}

Differentiating Eq.~\eqref{eq:KL} gives
\begin{equation}
  \frac{dF}{dt}=\sum_i \dot P_i(t)\ln\frac{P_i(t)}{\pi_i},
  \label{eq:dF_basic}
\end{equation}
where probability conservation has been used.  With the time-local form
\eqref{eq:time_local}, this becomes
\begin{equation}
  \frac{dF}{dt}=\sum_{i,j}(R_{\rm eff})_{ij}(t)P_j(t)
  \ln\frac{P_i(t)}{\pi_i}.
  \label{eq:dF_generator}
\end{equation}
This formula connects the sign of the entropy production rate to the structure
of the effective generator.

For a genuine Markov generator satisfying detailed balance, the right-hand side
is non-positive.  If some off-diagonal rates become negative, the usual proof of
monotonicity no longer applies.  This observation is the starting point for the
mechanism discussed below.

\section{Markovian monotonicity as a reference fact}

The following proposition is included to fix notation and to make clear which
part of the argument is standard.

\begin{proposition}[Relative entropy monotonicity for reversible Markov dynamics]
Let $R$ be an irreducible Markov generator with stationary distribution $\pi$.
If $R$ satisfies detailed balance,
\begin{equation}
  R_{ij}\pi_j=R_{ji}\pi_i,
\end{equation}
then the KL divergence $D_{\rm KL}(P(t)\|\pi)$ along the Markovian evolution
$\dot P=RP$ is monotonically non-increasing.
\end{proposition}

\begin{proof}
This is the standard entropy-dissipation property of reversible Markov chains.
Using detailed balance, one can write
\begin{equation}
  \frac{dF}{dt}
  =-\frac{1}{2}\sum_{i,j}R_{ij}\pi_j
  \left(\frac{P_j}{\pi_j}-\frac{P_i}{\pi_i}\right)
  \left[\ln\frac{P_j}{\pi_j}-\ln\frac{P_i}{\pi_i}\right].
\end{equation}
Each term is non-positive because $(x-y)(\ln x-\ln y)\geq 0$ for positive
$x,y$.  Hence $dF/dt\leq0$.
\end{proof}

The proposition is not presented as a new theorem.  Its role is to provide the
baseline against which non-Markovian overshoot is defined.  Entropy overshoot is
precisely a transient violation of this Markovian monotonicity property.

\section{Memory-induced breakdown of monotonicity}

We now turn to the non-Markovian case.  The essential point is that the effective
generator $R_{\rm eff}(t)$ of the observed process may fail to be a Markov
generator even when the full embedded dynamics is well behaved.

\begin{proposition}[Divisibility-loss mechanism for entropy overshoot]
Suppose that the effective time-local generator $R_{\rm eff}(t)$ develops a
negative off-diagonal element during a finite time interval.  Then the observed
evolution is not divisible into ordinary Markovian stochastic maps throughout
that interval.  Consequently, the Markovian relative-entropy contraction
argument cannot be applied, and transient intervals with $dF/dt>0$ may occur for
suitable initial conditions.
\end{proposition}

\begin{proof}[Proof sketch]
If $R_{\rm eff}(t)$ is a legitimate Markov generator for every intermediate
time, then the propagator between any two times is a stochastic map.  The
relative entropy is contractive under stochastic maps sharing the same stationary
state, and the standard entropy-dissipation argument applies.  A negative
off-diagonal element of $R_{\rm eff}(t)$ violates the generator condition for an
intermediate Markov step.  The evolution may still preserve positivity at the
level of the full memory dynamics, but the reduced observable dynamics is no
longer divisible into Markovian pieces.  Hence the contraction proof breaks down
on that interval.  Because the derivative of $F(t)$ is given by
Eq.~\eqref{eq:dF_generator}, the non-Markovian contribution generated by the
negative effective rate can have either sign, and for appropriate initial
directions it gives a positive entropy-rate interval.
\end{proof}

The proposition should not be read as a universal necessary-and-sufficient
criterion.  It identifies the mechanism that is relevant for the models studied
below.  Negative effective rates indicate loss of divisibility; once divisibility
is lost, the contraction property of relative entropy with respect to
intermediate stochastic maps no longer applies.  The KL divergence can then
temporarily increase, depending on the initial condition and the direction of the
memory-induced feedback.

This is the sense in which entropy overshoot provides a classical stochastic
analog of information backflow.  Information that had been effectively stored in
hidden memory variables can return to the observed degrees of freedom, producing
a temporary increase of distinguishability from the stationary state.

\section{A timescale criterion}

For exponential memory kernels the appearance of overshoot is governed mainly by
the competition between the intrinsic relaxation time and the memory time.  Let
$\lambda$ denote a characteristic relaxation rate of the Markov generator $R$,
for example the spectral gap in the reversible case.  The intrinsic relaxation
time is then $\lambda^{-1}$, while the memory time is $\gamma^{-1}$.

\begin{criterion}[Timescale criterion]
For the exponential kernel $K(t)=\gamma e^{-\gamma t}R$, entropy overshoot is
expected when the memory time is comparable to or longer than the intrinsic
relaxation time, i.e.
\begin{equation}
  \gamma \lesssim \lambda.
  \label{eq:criterion}
\end{equation}
\end{criterion}

\noindent\textit{Physical argument.}
If $\gamma\gg\lambda$, memory decays rapidly compared with the relaxation of the
system.  The dynamics is then close to an ordinary Markov process and the KL
divergence decreases monotonically.  If $\gamma\lesssim\lambda$, the memory
variable retains information on the timescale over which the observed system is
already relaxing.  Delayed feedback from this memory can drive the observable
probability vector away from the stationary distribution before final relaxation
sets in.  This produces entropy overshoot.  The numerical phase diagrams below
show that the boundary between the two regimes is well captured by this
timescale estimate.

The criterion is not meant to replace a model-specific stability analysis.  Its
purpose is to give a compact physical explanation of why overshoot is enhanced
by long memory.

\section{Minimal two-state model}

We first consider a two-state system with states $0$ and $1$.  The Markov
generator is
\begin{equation}
  R=\begin{pmatrix}
  -k_{01} & k_{10}\\
   k_{01} & -k_{10}
  \end{pmatrix}.
\end{equation}
The stationary distribution is
\begin{equation}
  \pi=\left(\frac{k_{10}}{k_{01}+k_{10}},
  \frac{k_{01}}{k_{01}+k_{10}}\right)^T.
\end{equation}
The Markovian relaxation time is
\begin{equation}
  \tau=(k_{01}+k_{10})^{-1}.
\end{equation}
For the Markovian equation $\dot P=RP$, the KL divergence decreases
monotonically.

We now introduce the exponential memory kernel
\begin{equation}
  K(t)=\gamma e^{-\gamma t}R.
\end{equation}
The resulting generalized master equation is
\begin{equation}
  \dot P(t)=\int_0^t \gamma e^{-\gamma(t-\tau)}RP(\tau)\,d\tau.
\end{equation}
When $\gamma$ is small compared with the Markovian relaxation rate, the memory
variable changes slowly and delayed feedback becomes significant.  The KL
divergence can then show a transient maximum before decaying to zero.

\begin{figure}[t]
\centering
\includegraphics[width=0.72\linewidth]{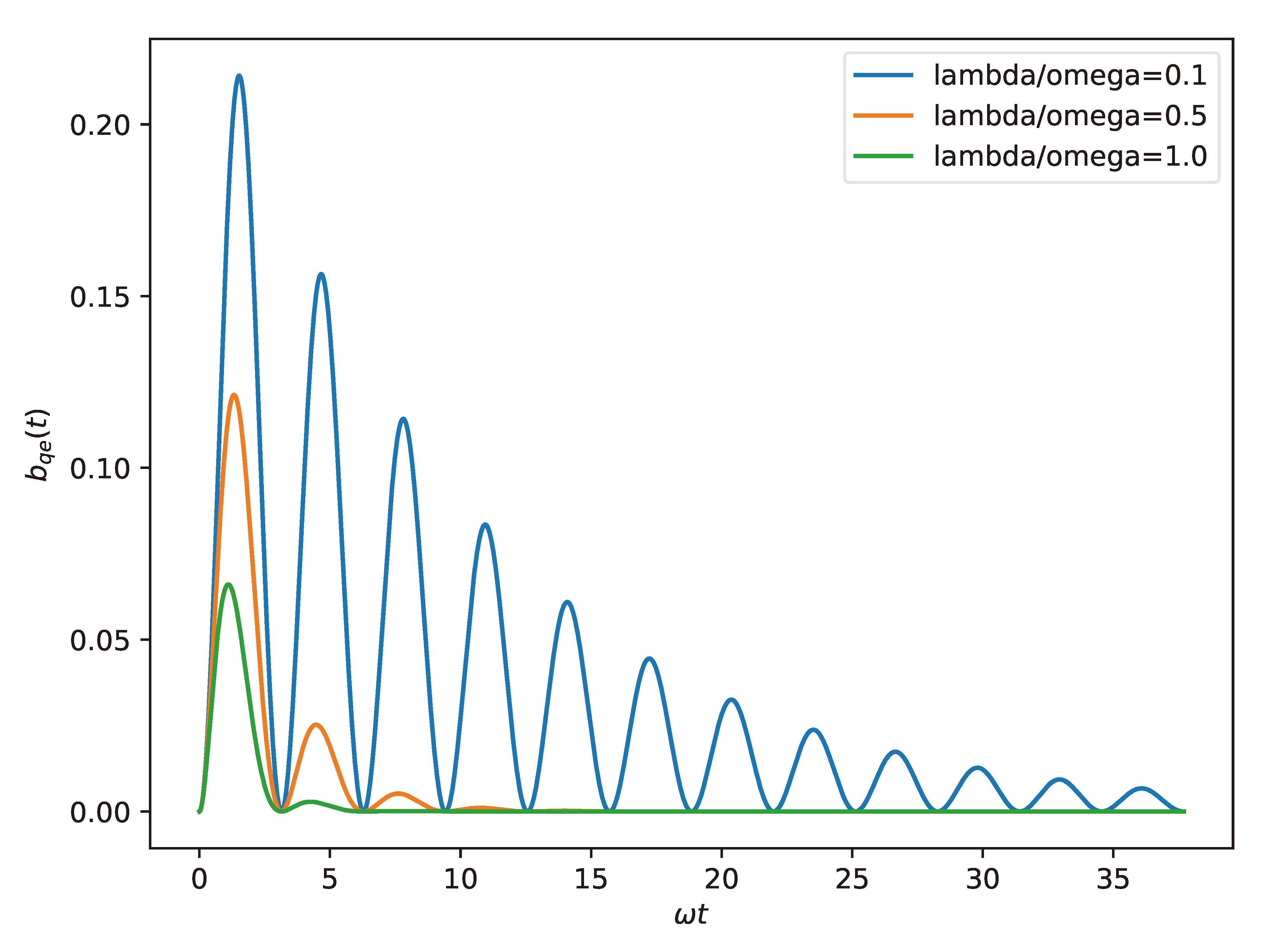}
\caption{Example of entropy overshoot in a minimal non-Markovian two-state
model.  The plotted quantity is a non-negative backflow/overshoot indicator
constructed from the positive part of the entropy-rate reversal.  Smaller
values of the memory parameter correspond to longer memory and stronger
oscillatory transient behavior.}
\label{fig:two_state}
\end{figure}

Figure~\ref{fig:two_state} illustrates the typical behavior.  The purpose of the
example is not to claim universality of the two-state model, but to show the
basic feedback mechanism in the simplest possible setting.

\section{Three-state model}

To obtain a richer phase structure, consider the symmetric three-state generator
\begin{equation}
  R=\begin{pmatrix}
  -2a & a & a\\
  a & -2a & a\\
  a & a & -2a
  \end{pmatrix}.
  \label{eq:three_state_R}
\end{equation}
The stationary distribution is uniform,
\begin{equation}
  \pi_i=\frac{1}{3}.
\end{equation}
With the exponential memory kernel
\begin{equation}
  K(t)=\gamma e^{-\gamma t}R,
\end{equation}
the system exhibits both monotonic and non-monotonic relaxation depending on the
ratio between $\gamma$ and $a$.

The use of a symmetric model is mainly for clarity.  The overshoot mechanism does
not require symmetry.  For a generic irreducible three-state generator with a
non-uniform stationary distribution, numerical integration gives the same
qualitative conclusion: sufficiently long memory can produce transient entropy
increase.  The detailed trajectory depends on the rates, but the controlling
factor remains the comparison between the memory time and the relaxation time.

\section{Phase diagrams}

To characterize the transition between relaxation regimes, we define the
overshoot amplitude
\begin{equation}
  \Delta F=\max_{t\geq0}F(t)-F(0).
  \label{eq:DeltaF}
\end{equation}
The monotonic regime has $\Delta F\leq0$, while entropy overshoot corresponds to
$\Delta F>0$.  In numerical plots we use the positive part of this quantity as
the displayed overshoot amplitude.

We also define a relaxation time $t_r$ by
\begin{equation}
  t_r=\min\{t:F(t)<\epsilon\},
\end{equation}
for a small threshold $\epsilon$.  Sweeping the parameters $(\gamma,a)$ produces
a phase diagram separating monotonic relaxation from entropy overshoot.

\begin{figure}[t]
\centering
\includegraphics[width=0.70\linewidth]{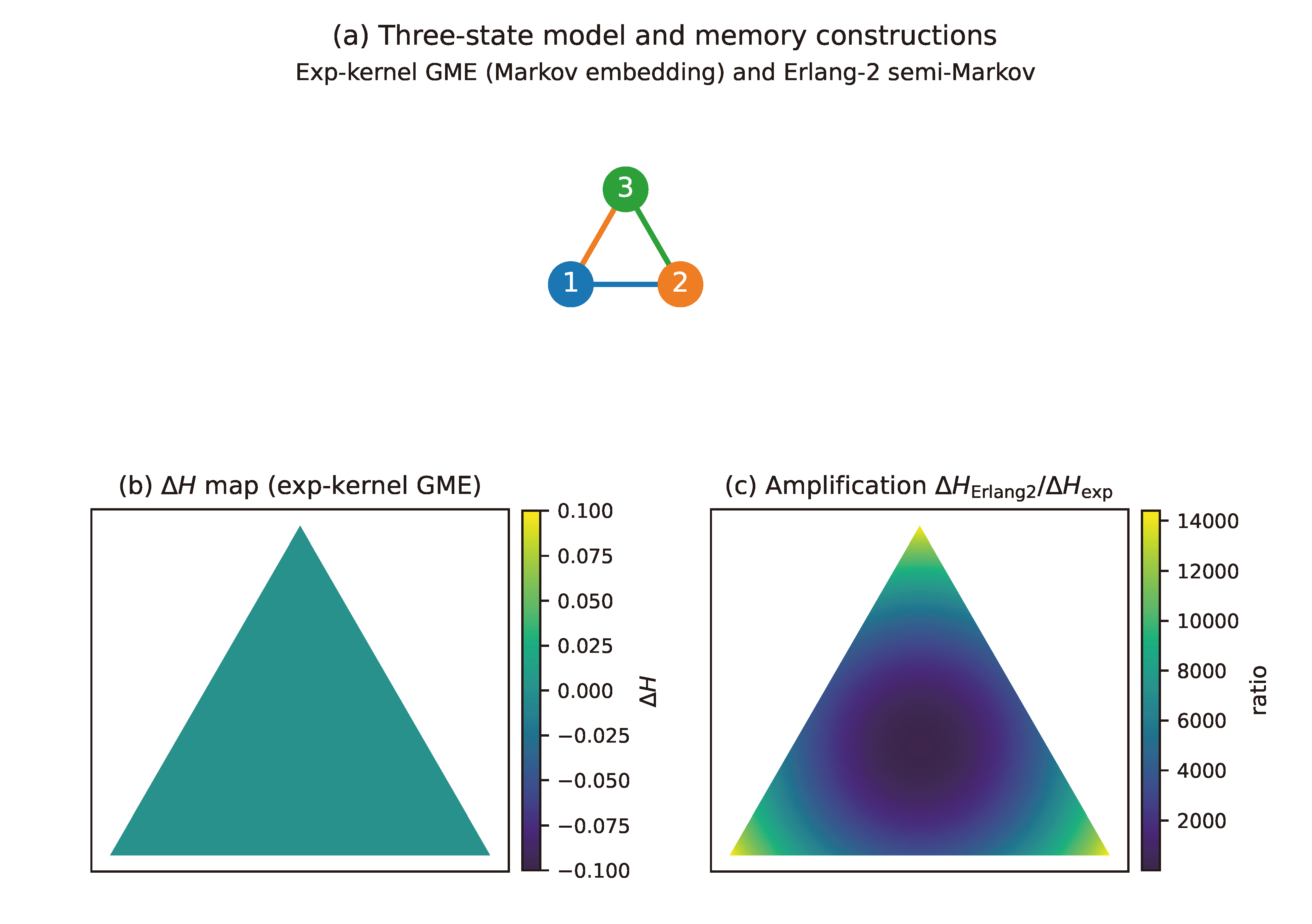}
\caption{Divisibility-loss phase diagram for the three-state model.  The
color scale indicates the strength of transient entropy amplification or an
associated amplification ratio.  The boundary separates a divisible relaxation
regime, where the KL divergence is monotonic, from a non-divisible regime where
entropy overshoot becomes possible.  The position of the boundary is controlled
primarily by the ratio between the memory time and the intrinsic relaxation
time.}
\label{fig:phase}
\end{figure}

Figure~\ref{fig:phase} summarizes the qualitative behavior.  The overshoot region
expands as the memory time increases.  This supports the timescale criterion
\eqref{eq:criterion} and shows that entropy overshoot is not a fine-tuned effect
of a single trajectory.

\section{Numerical method}

The generalized master equation was solved by two complementary methods.  First,
we directly integrated the integro-differential equation using a fourth-order
Runge--Kutta scheme combined with a trapezoidal discretization of the memory
integral.  The time step was decreased until the resulting KL-divergence curves
were insensitive to further refinement.

Second, for the exponential memory kernel, we used the embedded time-local
system~\eqref{eq:embed_P}--\eqref{eq:embed_Y}.  The enlarged system was integrated
directly and then projected onto the observed probability vector $P(t)$.  The
trajectories obtained by the two methods agreed within numerical accuracy.

For the phase diagrams, the parameter region was scanned over representative
ranges of $\gamma$ and $a$.  For each parameter pair the KL divergence trajectory
was computed and the overshoot amplitude~\eqref{eq:DeltaF} was evaluated.

\section{Results}

The numerical results support the structural mechanism described above.  When
memory is short, the dynamics closely follows ordinary Markovian relaxation and
the KL divergence decays monotonically.  When memory is long, delayed feedback
from the hidden memory variable can reverse the initial relaxation trend and
produce entropy overshoot.

In the two-state model, the overshoot is visible as an oscillatory transient in
the entropy-related backflow indicator.  The amplitude of this transient grows as
the memory time increases.  In the three-state model, parameter scans show a
clear separation between monotonic and overshoot regimes.  The observed boundary
is consistent with the estimate $\gamma\sim\lambda$.

These results indicate that entropy overshoot is a robust consequence of
memory-induced non-divisibility rather than an artifact of a special initial
condition.  The minimal models provide simple examples, while the mechanism
itself is formulated in terms of the effective generator and therefore applies
more broadly.

\section{Discussion}

The present work gives a classical stochastic interpretation of information
backflow.  In non-Markovian open quantum systems, information backflow is often
associated with loss of divisibility of the dynamical map
\cite{Breuer2009,Rivas2014}.  The same structural idea appears here in a
finite-state classical setting: memory can make the effective time-local
generator non-divisible, and this allows a temporary increase of the KL
divergence from the stationary distribution.

It is important to distinguish this conclusion from previous uses of
relative-entropy revivals or negative entropy-production rates as witnesses of
non-Markovianity.  Earlier work clarified classical versus quantum notions of
non-Markovianity \cite{Vacchini2012} and showed that negative entropy-production
rates can be associated with memory effects in several open-system settings
\cite{Bhattacharya2017,Popovic2018,Strasberg2019}.  The present paper does not
claim that relative-entropy revivals themselves are new.  Its contribution is to
identify a minimal finite-state mechanism that produces such revivals: a memory
kernel generates transient negative effective rates, these rates signal loss of
stochastic divisibility, and the relative-entropy contraction argument is then
locally inapplicable.

This interpretation also explains why the phenomenon is absent in ordinary
reversible Markov chains.  In that case the relative entropy is contractive and
serves as a Lyapunov function.  Memory changes the situation not by changing the
stationary distribution, but by changing the transient direction of the observed
flow in probability space.  Related non-Markovian stochastic formulations have
also been investigated in alternative approaches such as stochastic wave-function
methods and Gaussian non-Markovian processes \cite{Budini2004}.

The main limitation of the present work is that the timescale criterion is a
structural and diagnostic criterion rather than a universal necessary-and-
sufficient condition.  Different kernels and different state-space structures
can modify the detailed form of the phase boundary.  Nevertheless, the minimal
models show that the basic mechanism is generic: entropy overshoot occurs when
memory persists long enough to feed information back into the observed degrees of
freedom.

This viewpoint suggests several extensions.  One may consider continuous-state
systems, non-exponential kernels, semi-Markov waiting-time distributions, or
quantum stochastic dynamics.  Another direction is to compare different
information measures and identify which transient features are metric-dependent
and which are controlled by the divisibility structure alone.

\section{Conclusion}

We have analyzed entropy overshoot in finite-state non-Markovian stochastic
relaxation.  The phenomenon is defined as a transient increase of the KL
divergence from the stationary distribution during a relaxation process.

The main contribution of this work is not the observation of relative-entropy
revival itself, which has been discussed previously in connection with
non-Markovianity and entropy-production rates.  Rather, the contribution is to
identify memory-induced loss of divisibility as a minimal structural mechanism
underlying entropy overshoot.

When the effective time-local generator remains Markovian, the usual
monotonicity of relative entropy is retained.  When memory produces transient
negative effective rates, the reduced observable dynamics ceases to be divisible
into intermediate Markov steps.  The contraction proof for the KL divergence is
then no longer available, and entropy overshoot can occur.

Minimal two- and three-state models demonstrate the mechanism explicitly and
show that the transition between monotonic and non-monotonic relaxation is
controlled by the ratio between the memory time and the intrinsic relaxation
time.  The resulting divisibility-loss phase diagrams provide a compact
representation of the overshoot regime.

The results offer a simple theoretical framework for classical information
backflow in stochastic systems with memory and provide a basis for further study
of non-Markovian relaxation in more complex settings.


\begin{thebibliography}{99}

\bibitem{Nakajima1958}
S. Nakajima, On quantum theory of transport phenomena, Prog. Theor. Phys. 20,
948 (1958).

\bibitem{Zwanzig2001}
R. Zwanzig, \emph{Nonequilibrium Statistical Mechanics} (Oxford University
Press, 2001).

\bibitem{Kubo1966}
R. Kubo, The fluctuation-dissipation theorem, Rep. Prog. Phys. 29, 255 (1966).

\bibitem{Gardiner2004}
C. Gardiner, \emph{Handbook of Stochastic Methods} (Springer, 2004).

\bibitem{BreuerPetruccione2002}
H.-P. Breuer and F. Petruccione, \emph{The Theory of Open Quantum Systems}
(Oxford University Press, 2002).

\bibitem{Breuer2009}
H.-P. Breuer, E.-M. Laine, and J. Piilo, Measure for the degree of
non-Markovian behavior of quantum processes in open systems, Phys. Rev. Lett.
103, 210401 (2009).

\bibitem{Rivas2014}
A. Rivas, S. F. Huelga, and M. B. Plenio, Quantum non-Markovianity:
characterization, quantification and detection, Rep. Prog. Phys. 77, 094001
(2014).

\bibitem{Vacchini2012}
B. Vacchini, A classical appraisal of quantum definitions of non-Markovian
dynamics, J. Phys. B: At. Mol. Opt. Phys. 45, 154007 (2012).

\bibitem{Bhattacharya2017}
S. Bhattacharya, A. Misra, C. Mukhopadhyay, and A. K. Pati, Exact master
equation for a spin interacting with a spin bath: Non-Markovianity and negative
entropy production rate, Phys. Rev. A 95, 012122 (2017).

\bibitem{Popovic2018}
M. Popovic, B. Vacchini, and S. Campbell, Entropy production and correlations
in a controlled non-Markovian setting, Phys. Rev. A 98, 012130 (2018).

\bibitem{Strasberg2019}
P. Strasberg and M. Esposito, Non-Markovianity and negative entropy production
rates, Phys. Rev. E 99, 012120 (2019).

\bibitem{Magnetization2026}
F. Hartmann, F. Tietjen, R. M. Geilhufe, and J. Anders, Quantifying
non-Markovianity in magnetization dynamics via entropy production rates,
arXiv:2602.17384 (2026).

\bibitem{CoverThomas2006}
T. Cover and J. Thomas, \emph{Elements of Information Theory} (Wiley, 2006).

\bibitem{Qian2001}
H. Qian, Relative entropy: free energy associated with equilibrium fluctuations
and nonequilibrium deviations, J. Phys. Chem. B 105, 10733 (2001).

\bibitem{Seifert2012}
U. Seifert, Stochastic thermodynamics, fluctuation theorems and molecular
machines, Rep. Prog. Phys. 75, 126001 (2012).

\bibitem{VanKampen2007}
N. G. van Kampen, \emph{Stochastic Processes in Physics and Chemistry}
(North-Holland, 2007).

\bibitem{HanggiJung1995}
P. H{\"a}nggi and P. Jung, Colored noise in dynamical systems, Adv. Chem. Phys.
89, 239 (1995).

\bibitem{MetzlerKlafter2000}
R. Metzler and J. Klafter, The random walk's guide to anomalous diffusion: a
fractional dynamics approach, Phys. Rep. 339, 1 (2000).

\bibitem{MontrollWeiss1965}
E. Montroll and G. Weiss, Random walks on lattices. II, J. Math. Phys. 6, 167
(1965).

\bibitem{Budini2004}
A. A. Budini, Non-Markovian Gaussian dissipative stochastic wave vector, Phys.
Rev. A 69, 042107 (2004).

\end{thebibliography}
\end{document}